\documentclass{article}

\usepackage{amsmath,amssymb,pifont}
\usepackage{multicol}
\usepackage{amstext}
\usepackage{amsthm}
\usepackage{multirow}
\usepackage{booktabs}
\usepackage[skip=0pt]{subcaption}
\usepackage{times}
\usepackage{lipsum}
\usepackage[shortlabels]{enumitem}
\usepackage{cancel}
\usepackage{wrapfig}
\usepackage{array}
\usepackage{siunitx}
\usepackage{csvsimple}
\usepackage[multidot]{grffile}
\usepackage{bbm}
\usepackage{dblfloatfix}
\usepackage{hyperref}
\usepackage{makecell}
\usepackage{bbm, dsfont}
\usepackage{mathtools}
\usepackage[dvipsnames]{xcolor}
\usepackage{comment}
\usepackage[capitalise]{cleveref}

\usepackage{geometry}
\usepackage{enumitem}

\usepackage[multiple]{footmisc}
\usepackage{mathrsfs}
\usepackage{todonotes}
\usepackage{tikz}
\usepackage{hyperref}
\usepackage{cleveref}
\usepackage{algpseudocode,algorithm,algorithmicx}
\usepackage{xfrac}

\usepackage{graphicx} 
\usepackage{color}

\usepackage{array}
\usepackage{amssymb}
\usepackage{amsmath}
\usepackage{xspace}
\usepackage{fancyhdr}
\usepackage{comment}

\hypersetup{final}

\newcommand{\boldzero}{\ensuremath{\boldsymbol{0}}}

\newcommand{\bfu}{\ensuremath{\mathbf{u}}}
\newcommand{\bfv}{\ensuremath{\mathbf{v}}}
\newcommand{\bfw}{\ensuremath{\mathbf{w}}}
\newcommand{\bfx}{\ensuremath{\mathbf{x}}}

\newcommand{\bfz}{\ensuremath{\mathbf{z}}}

\newcommand{\calM}{\ensuremath{\mathcal{M}}}

\renewcommand{\Pr}{\mathop{\mathbf{Pr}}}

\newtheorem{lemma}{Lemma}[section]
\newtheorem{theorem}[lemma]{Theorem}

\newtheorem{definition}[lemma]{Definition}

\DeclareMathOperator*{\Bern}{Bern}
\DeclareMathOperator*{\Binom}{Binom}

\makeatletter
\newcommand{\vast}{\bBigg@{4}}
\newcommand{\Vast}{\bBigg@{5}}
\makeatother

\newcommand{\ex}[2]{{\ifx&#1& \mathbb{E} \else
\underset{#1}{\mathbb{E}} \fi \left[#2\right]}}
\newcommand{\pr}[2]{{\ifx&#1& \mathbb{P} \else
\underset{#1}{\mathbb{P}} \fi \left[#2\right]}}

\newcommand{\ltwo}[1]{\left\|#1\right\|_2}

\usepackage{ulem}

\DeclarePairedDelimiterX{\infdivx}[2]{(}{)}{%
  #1\;\delimsize\|\;#2%
}

\newcommand{\mypar}[1]{\smallskip
	\noindent{\textbf{{#1}:}}}
	
\renewcommand{\epsilon}{\varepsilon}

\setlist{nolistsep}
\setlist[itemize]{noitemsep, topsep=0pt}

\setlist{nolistsep}
\setlist[itemize]{noitemsep, topsep=0pt}

\usepackage{fullpage}
\usepackage[numbers, sort, comma, square]{natbib}

\begin{document}

\title{Tighter Privacy Analysis for Truncated Poisson Sampling}
\author{Arun Ganesh\thanks{\texttt{arunganesh@google.com}, Google Research}}
\maketitle

\begin{abstract}
We give a new privacy amplification analysis for truncated Poisson sampling, a Poisson sampling variant that truncates a batch if it exceeds a given maximum batch size.
\end{abstract}

\section{Introduction}

Privacy amplification by sampling is a standard technique for improving privacy-utility tradeoffs in private training algorithms such as DP-SGD. The most commonly analyzed version of privacy amplification by sampling is done by forming batches via Poisson sampling. In practice, Poisson sampling's variable batch sizes are in conflict with the fixed batch sizes needed for XLA compilation, which can be handled by padding or truncating the batch as discussed in \cite{chua2024scalable}. While padding with zero-gradient examples doesn't affect the sum of the gradients and hence doesn't affect privacy analysis, truncating a batch to a maximum size of $B$ after Poisson sampling does affect the sum of the gradients, and hence it could negatively impact the privacy guarantees and must be handled carefully.

One strategy proposed in \cite{chua2024scalable} to handle truncation is to bound the probability $\eta$ across all iterations that Poisson sampling exceeds a batch size of $B$, and absorb this bound into the $\delta$ term. Specifically, they show that if Poisson sampling satisfies $(\epsilon, \delta)$-DP then truncated Poisson sampling satisfies $(\epsilon, \delta + e^\epsilon \eta)$-DP. However, to make $q$ sufficiently small requires us to sample less than $B$ examples in expectation, i.e. there is some amount of wasted computation. The goal of this note is to provide a less lossy privacy guarantee for a specific instatiation of truncated Poisson sampling and reduce the amount of wasted computation in practice.

\section{Problem Statement}

For this note we will consider datasets consisting of vectors in the $\ell_2$-unit ball $\mathbb{B}(\boldzero, 1)$. That is, $D = \{\bfx_1, \bfx_2, \ldots, \bfx_n\}$ is a dataset where $\bfx_i \in \mathbb{R}^d$ and $\ltwo{\bfx_i} \leq 1$ for all $i$. The choice of norm 1 is without loss of generality and can be easily replaced with other values. By standard techniques such as post-processing and adaptive composition, this lets our analysis readily generalize to adaptive vector queries on more general datasets, e.g. computing batch gradients on a dataset.

In \cref{fig:tps} we define the mechanism $\calM_{B,p,\sigma}$ of interest, a vector sum on a subsample of a vector dataset $D \in \mathbb{B}(\boldzero, 1)^n$. Specifically, we independently include each example in $D$ with probability $p$, letting $S$ be the set of included examples. If $|S| \leq B$ then we output $S$, otherwise we take a uniformly random size-$B$ subset of $S$. Then we return the sum of the vectors in $S$, plus Gaussian noise with standard deviation $\sigma$.

\begin{figure}[H]
\begin{algorithm}[H]
\caption{$\calM_{B, p, \sigma}$: Truncated Poisson sampled vector sum}
\label{fig:tps}
\textbf{Parameters:} Truncated batch size $B$, sampling probability $p$.
\newline\textbf{Inputs:} Dataset $D = \{\bfx_1, \bfx_2, \ldots, \bfx_n\} \in \mathbb{B}(\boldzero, 1)^n$
\begin{algorithmic}[1]
\State $S \leftarrow \emptyset$
\For{$i \in [n]$}
\State $S \leftarrow S \cup \{\bfx_i\}$ w.p. $p$, $S \leftarrow S$ otherwise.
\EndFor
\If{$|S| > B$}
\State $S \leftarrow$ uniformly random subset of $S$ of size $B$.
\EndIf
\State $\bfz \sim N(0, \sigma^2 \mathbb{I})$
\State \Return $\sum_{\bfx \in S} \bfx + \bfz$.
\end{algorithmic}
\end{algorithm}
\end{figure}

It is important to note that we do truncation \textit{randomly} rather than arbitrarily. This gives a slightly better privacy guarantee, and we believe it will usually be reasonable to do so in the random-access settings where one would implement Poisson sampling anyway. 

We next define several adjacency definitions of interest.

\begin{definition}

We define the following adjacencies:
\begin{itemize}
    \item \textbf{Add-or-remove-1-of-$n$:} We say $D \stackrel{AR,n}{\sim} D'$ if $|D| = n$ and $D'$ is a size $n-1$ subset of $D$.
    \item \textbf{Zero-out} We say $D \stackrel{ZO,n}{\sim} D'$ if $|D| = |D'| = n$ and $D$ is equal to $D'$ with one element replaced with $\boldzero$.
    \item \textbf{Replace-one} We say $D \stackrel{RO,n}{\sim} D'$ if $|D| = |D'| = n$ and $D$ is equal to $D'$ with one element replaced with an arbitrary element of $\mathbb{B}(\boldzero, 1)$.
\end{itemize}
\end{definition}

\subsection{Dominating Pairs}

It will be convenient to state our results in terms of dominating pairs. Dominating pairs have many equivalent definitions, we use the following one:

\begin{definition}
Two distributions $P, Q$ are \textbf{$(\epsilon, \delta)$-indistinguishable} for $\epsilon \in \mathbb{R}, \delta \in [0, 1]$ if:
\[\forall S: P(S) \leq e^\epsilon Q(S) + \delta.\]
\end{definition}

\begin{definition}
$P, Q$ is a \textbf{dominating pair} for $A, B$ if the following is true: if $P, Q$ are $(\epsilon, \delta)$-indistinguishable, then $A, B$ are $(\epsilon, \delta)$-indistinguishable.

$P, Q$ is a dominating pair for $\calM$ (with respect to an adjacency definition $\sim$) if for all $D \sim D'$, $P, Q$ is a dominating pair for $\calM(D), \calM(D')$.
\end{definition}

We note that as above, dominating pairs are stated in a ``one-directional'' manner, i.e. the indistinguishability definition only guarantees $P(S) \leq e^\epsilon Q(S) + \delta$ but not $Q(S) \leq e^\epsilon P(S) + \delta$. However, \cite{pmlr-v151-zhu22c} show that if $P, Q$ is a dominating pair in one direction, then $Q, P$ is a dominating pair in the other direction. So we will give a dominating pair in one direction only for convenience, and implicitly our results give a dominating pair in the other direction.

Hence to prove a privacy guarantee for $\calM$ under the standard notion of $(\epsilon, \delta)$-DP, it suffices to state a dominating pair whose $(\epsilon, \delta)$-indistinguishability is readily verifiable, which will be the main goal of this note.

The following lemma, which can be proven as a corollary of Theorem 3.4 of \cite{schuchardt2024unified}, will be useful to us:

\begin{lemma}\label{lem:dr}
For any sets of vectors $U = \{\bfu_1, \bfu_2, \ldots, \bfu_n\}$, $V = \{\bfv_1, \bfv_2, \ldots \bfv_m\}$ and associated probabilities $p_1, \ldots, p_n \geq 0: \sum_{i \in [n]} p_i = 1$, $q_1, \ldots, q_n \geq 0: \sum_{j \in [m]} q_j = 1$, 

\[\sum_{i \in [n]} p_i \cdot N(-\ltwo{\bfu_i}, \sigma^2), \sum_{j \in [m]} q_j \cdot N(\ltwo{\bfv_j}, \sigma^2)\] 

is a dominating pair for

\[\sum_{i \in [n]} p_i \cdot N(\bfu_i, \sigma^2 \mathbb{I}), \sum_{j \in [m]} q_j \cdot N(\bfv_j, \sigma^2 \mathbb{I}).\]
\end{lemma}

\section{Our Analysis}

\subsection{Add-Remove Adjacency}

For ease of exposition we start with the simplest version of our result, which uses a restricted version of the add-or-remove-1-of-$n$ adjacency $\stackrel{AR,n}{\sim}$. Normally, the add-remove adjacency says $D$ and $D'$ are adjacent if $D'$ is a size $|D|-1$ subset of $D$ or vice-versa, without a restriction on $|D|$. Our privacy guarantee depends on $|D|$, so we need to restrict to $D, D'$ pairs of a specific size $n$; we discuss this further after the result.

\begin{theorem}\label{thm:truncation-addremove}
Given $n, p, B, \sigma$ define the distributions $P, Q$ as follows:
\begin{itemize}
    \item We sample 
    $r = 1$ w.p. $ \Pr[\Binom(n-1, p) < B]$ and $r = 2$ otherwise.
    \item For $P$, we set 
    \[x = r \cdot \Bern(q), q = \left\{\begin{array}{lr}
        p & \text{if } r = 1,\\
        p \cdot \sum_{s = B}^{n-1} \frac{\Pr[\Binom(n-1, p) = s]}{\Pr[\Binom(n-1, p) \geq B]} \cdot \frac{B}{s+1} & \text{if } r = 2.\\
    \end{array}\right\}\]
    For $Q$, we set $x = 0$.
    \item The final random variable is $(r, x + z)$, where $z \sim N(0, \sigma^2)$.
\end{itemize}.

$P, Q$ is a dominating pair for $\calM_{B,p,\sigma}$ with respect to $\stackrel{AR,n}{\sim}$.
\end{theorem}

Before we give the proof; we note that the expression $p \cdot \sum_{s = B}^{n-1} \frac{\Pr[\Binom(n-1, p) = s]}{\Pr[\Binom(n-1, p) \geq B]} \cdot \frac{B}{s+1}$ is somewhat cumbersome to evaluate in practice. In \cref{sec:implementation} we give an equivalent formula that is easier to evaluate and which we recommend using instead in practice.

\begin{proof}
Without loss of generality, let $D = \{\bfx_1, \bfx_2, \ldots, \bfx_n\}$, $D' = \{\bfx_2, \ldots, \bfx_n\}$ be two datasets that are adjacent under the add-or-remove-one adjacency.

We will instead consider the following equivalent process for choosing $S$ in \cref{fig:tps} from $D$ or $D'$:

\begin{itemize}
    \item We let $\pi_{n-1}$ be a uniformly random permutation of the examples $\{\bfx_2, \bfx_3, \ldots, \bfx_n\}$.
    \item We sample $s_{n-1} \sim \Binom(n-1, p)$ and $s_1 \sim \Bern(p)$.
    \item \textbf{If $s_{n-1} < B$:} We return the first $s_{n-1}$ elements in $\pi_{n-1}$. We additionally return  $\bfx_1$ if $s_1 = 1$ and we are sampling from $D$.
    \item \textbf{If $s_{n-1} \geq B$:} \begin{itemize}
        \item We sample $s_1' \sim \Bern(\frac{B}{s_{n-1} + 1})$.
        \item If $s_1 \cdot s_1' = 1$ and we are sampling from $D$, we return the first $B - 1$ elements of $\pi_{n-1}$ and $\bfx_1$. 
        \item Otherwise we return the first $B$ elements of $\pi_{n-1}$.
    \end{itemize}
\end{itemize}

\begin{lemma}\label{lem:equivalent}
The above procedure arrives at the same distribution of $S$ as \cref{fig:tps}.
\end{lemma}
\begin{proof}
We can see this is equivalent by considering three events when sampling from $D$:
\begin{itemize}
    \item $E_1$: At most $B-1$ elements are sampled from $D \setminus \{\bfx_1\}$ prior to truncation.
    \item $E_2$: At least $B$ elements are sampled from $D \setminus \{\bfx_1\}$ prior to truncation, and $\bfx_1$ is not sampled (also prior to truncation).
    \item $E_3$: At least $B$ elements are sampled from $D \setminus \{\bfx_1\}$ prior to truncation, and $\bfx_1$ is also sampled.
\end{itemize}

$E_1$ happens with probability $\Pr[s_{n-1} < B]$, and it is easy to see the above procedure conditioned on $s_{n-1} < B$ has the same distribution as truncated Poisson sampling conditioned on $E_1$.

Similarly, $E_2$ happens with probability $\Pr[s_{n-1} \geq B \land s_1 = 0]$, and it is again easy to see that the above procedure conditioned on  $s_{n-1} \geq B \land s_1 = 0$ has the same distribution as truncated Poisson sampling conditioned on $E_2$.

For $E_3$ and the remaining event $s_{n-1} \geq B \land s_1 = 1$ for the above procedure, if $s_{n-1}$ samples from $D \setminus \{\bfx_1\}$ and $\bfx_1$ are sampled prior to truncation, then $\bfx_1$ and $B-1$ elements from $D \setminus \bfx_{n-1}$ are included in the final sample w.p. $\frac{B}{s_{n-1} + 1}$ (and this corresponds to the case $s_1' = 1$), and otherwise $B$ elements from $D \setminus \bfx_{n-1}$ are included (corresponding to the case $s_1' = 0$). So we can also verify the distributions match conditioned on $E_3$ and the event $s_{n-1} \geq B \land s_1 = 1$ respectively.

For sampling from $D'$, the argument is the same except that $E_2, E_3$ can be merged into a single event.
\end{proof}

By the post-processing inequality, we can assume that in addition to releasing the output of $\calM_{B, p, \sigma}$, we release $\pi_{n-1}$ and $\max\{s_{n-1}, B\}$. Given this information:

\begin{itemize}
    \item If $\max\{s_{n-1}, B\} < B$, conditioned on $\pi_{n-1}$ and $\max\{s_{n-1}, B\}$ the distribution of the sum of the vectors in $S$ is the sum of (i) the first $s_{n-1}$ vectors in $\pi_{n-1}$ (ii) $\Bern(p)$ times $\bfx_1$ if sampling from $D$, $\boldzero$ if sampling from $D'$, and (iii) the Gaussian noise. (i) is a constant known to the adversary, so this is just a Poisson subsampled Gaussian mechanism with sampling probability $p$ and noise multiplier $\sigma$. A dominating pair for this mechanism is $P | r = 1, Q | r = 1$ as defined in the theorem statement (this can be proven e.g. using \cref{lem:dr}). 
    \item If $\max\{s_{n-1}, B\} = B$, conditioned on $\pi_{n-1}$ and $\max\{s_{n-1}, B\}$ the distribution of the sum of the gradients is the sum of (i) the first $B$ examples in $\pi_{n-1}$, (ii) $\Bern(p)$ times $\Bern(\frac{B}{s_{n-1} + 1})$ times the difference between $\bfx_1$ and the $B$th-element of $\pi_{n-1}$ if sampling from $D$, $\boldzero$ if sampling from $D'$, and (iii) the Gaussian noise.
    
    Conditioned on $s_{n-1} \geq B$, $\Bern(p)$ times $\Bern(\frac{B}{s_{n-1} + 1})$ is equivalent to 
    
    \[\Bern\left(p \cdot \sum_{s = B}^{n-1} \frac{\Pr[\Binom(n-1, p) = s]}{\Pr[\Binom(n-1, p) \geq B]} \cdot \frac{B}{s+1}\right)\]
    
    So this is also just a Poisson subsampled Gaussian mechanism, but the sensitivity is doubled since (ii) is the difference between two gradients instead of a single gradient. A dominating pair for this mechanism is $P | r = 2, Q | r = 2$ as defined in the theorem statement. 
\end{itemize}

So, the combination of these two Poisson subsampled Gaussian mechanisms is the same as the mechanism described in the theorem statement, with releasing $r$ being equivalent to releasing $\max\{s_{n-1}, B\}$ as knowing the specific value of $s_{n-1} < B$ does not affect the privacy guarantee.
\end{proof}

Some comments are in order about the result:

\mypar{Privacy of $n$} Our adjacency definition and privacy guarantee notably depends on $n$, unlike Poisson sampling without truncation where the privacy guarantee and the add-remove adjacency definition are independent of $n$. In settings where we are only running repeated applications of $\calM_{B,n,\sigma}$ (e.g., a single run of DP-SGD), this dependence might be acceptable, as it still allows the data curator to make the statement ``the output of using this dataset is indistinguishable from the output obtained by deleting your example from the dataset.'' However, if we are composing $\calM_{B,n,\sigma}$ with other mechanisms, we need to be careful about using this result. Hence, we recommend avoiding the use of this result in these settings, and even when using this result just to run $\calM_{B,n,\sigma}$ with no other mechanisms, one should be precise when stating the formal privacy guarantees.

Note that for the zero-out and replace-one adjacencies that we discuss later, $n$ is public since adjacent datasets are the same size so the privacy guarantee depending on $n$ is not an issue.

\mypar{Tightness of the bound} If $B = \infty$ (i.e. truncation never happens), \cref{thm:truncation-addremove} retrieves the privacy guarantees of Poisson sampling without truncation, as desired. If we are in a setting where $p = 1, B < n$ (i.e. truncation always happens), truncated Poisson sampling retrieves the method of sampling a batch of fixed size $B$. This method has the privacy guarantees of a Poisson sampled Gaussian mechanism with doubled sensitivity and sampling probability $B/n$ (and this is tight), which is exactly what \cref{thm:truncation-addremove} recovers. So at the ``extremes'' of never truncating and always truncating, this bound is tight, although we conjecture it is possible to slightly improve this bound in the intermediate regime via an analysis that does not release $s_{n-1}$. Note that we only pay for the doubled sensitivity when truncation occurs; if $\Pr[\Binom(n-1, p) \geq B]$ is small, this bound nearly retrieves the privacy guarantees of Poisson sampling without truncation.

\mypar{Extension to BandMF} BandMF with cyclic Poisson sampling \cite{choquette2024amplified} is a generalization of DP-SGD with Poisson sampling, whose privacy can be analyzed via a reduction to the privacy of DP-SGD. Using this reduction, the above analysis (and the analogous results for other adjacencies) also extend immediately to BandMF with truncated cyclic Poisson sampling, by replacing $n$ in \cref{thm:truncation-addremove} with the size of each subset in the cyclic Poisson sampling scheme.

\subsection{Zero-Out Adjacency}

We next extend the results to the zero-out adjacency:

\begin{theorem}\label{thm:truncation-zeroout}
Given $n, p, B, \sigma$ define the distributions $P, Q$ as follows:
\begin{itemize}
    \item We sample 
    $r = 1$ w.p. $ \Pr[\Binom(n-1, p) < B]$ and $r = 2$ otherwise.
    \item For $P$, we set 
    \[x = \left\{\begin{array}{lr}
        \Bern(p) & \text{if } r = 1,\\
        2 \cdot \Bern\left(p \cdot \sum_{s = B}^{n-1} \frac{\Pr[\Binom(n-1, p) = s]}{\Pr[\Binom(n-1, p) \geq B]} \cdot \frac{B}{s+1}\right) & \text{if } r = 2.\\
    \end{array}\right\}.\]
    For $Q$, we set 
    \[x = \left\{\begin{array}{lr}
        0 & \text{if } r = 1,\\
        -\Bern\left(p \cdot \sum_{s = B}^{n-1} \frac{\Pr[\Binom(n-1, p) = s]}{\Pr[\Binom(n-1, p) \geq B]} \cdot \frac{B}{s+1}\right) & \text{if } r = 2.\\
    \end{array}\right\}.\]
    \item The final random variable is $(r, x + z)$, where $z \sim N(0, \sigma^2)$.
\end{itemize}

$P, Q$ is a dominating pair for $\calM_{B,p,\sigma}$ with respect to $\stackrel{ZO,n}{\sim}$.
\end{theorem}
\begin{proof}
The proof proceeds similarly to \cref{thm:truncation-addremove}. We use the sampling procedure for sampling $D$ in that proof to sample from both $D$ and $D'$, since $D'$ is now also size $n$. Then, defining $s_{n-1}$ to be the number of elements sampled from $D \setminus \{\bfx_1\}$ (or analogously $D' \setminus \{\bot\}$ as before):

\begin{itemize}
    \item If $\max\{s_{n-1}, B\} < B$, as in the add-remove case we reduce to a Poisson subsampled Gaussian mechanism with sampling probability $p$ and noise multiplier $\sigma$ which matches the $r = 1$ case in \cref{thm:truncation-addremove}.
    \item If $\max\{s_{n-1}, B\} = B$, we proceed differently. Now conditioned on $\pi_{n-1}$ and $\max\{s_{n-1}, B\}$ the distribution of the sum of the gradients is the sum of (i) the gradients of the first $B$ examples in $\pi_{n-1}$, (ii) $\Bern(p)$ times $\Bern(\frac{B}{s_{n-1} + 1})$ times the difference between $\bfx_1$ and the $B$th-element of $\pi_{n-1}$ if sampling from $D$, $\Bern(\frac{B}{s_{n-1} + 1})$ times $-1$ times the $B$th-element of $\pi_{n-1}$ if sampling from $D'$, and (iii) the Gaussian noise. This is the same as the add-remove case, except for the value of (ii) when sampling from $D'$.
    
    Again (i) is a constant known by the adversary, so this is a Gaussian mechanism applied to the random variable (ii). Conditioned on $s_{n-1} \geq B$, $\Bern(p)$ times $\Bern(\frac{B}{s_{n-1} + 1})$ is still equivalent to 
    
    \[\Bern\left(q\right), q = p \cdot \sum_{s = B}^{n-1} \frac{\Pr[\Binom(n-1, p) = s]}{\Pr[\Binom(n-1, p) \geq B]} \cdot \frac{B}{s+1}.\]
    
    Let $\pi_{n-1}(B)$ denote the $B$th element of $\pi_{n-1}$. Then given $D$ and conditioned on $s_{n-1} \geq B$, we are sampling from $(1-q) \cdot N(-\pi_{n-1}(B), \sigma^2 \mathbb{I}) + q \cdot N(\bfx_1, \sigma^2 \mathbb{I})$, and given $D'$ we are sampling from  $(1-q) \cdot N(-\pi_{n-1}(B), \sigma^2 \mathbb{I}) + q \cdot N(\boldzero, \sigma^2 \mathbb{I})$. By a linear translation distinguishing these distributions is equivalent to distinguishing $(1-q) \cdot N(\boldzero, \sigma^2 \mathbb{I}) + q \cdot N(\bfx_1-\pi_{n-1}(B), \sigma^2 \mathbb{I})$ and $(1-q) \cdot N(\boldzero, \sigma^2 \mathbb{I}) + q \cdot N(- \pi_{n-1}(B), \sigma^2 \mathbb{I})$. Since $\ltwo{\pi_{n-1}(B)} \leq 1, \ltwo{\bfx_1 - \pi_{n-1}(B)} \leq 2$, \cref{lem:dr} gives that the pair $(1-q) \cdot N(0, \sigma^2) + q \cdot N(2, \sigma^2)$ and $(1-q) \cdot N(0, \sigma^2) + q \cdot N(-1, \sigma^2)$ is a dominating pair for the mechanism conditioned on $s_{n-1} \geq B$, which is exactly the $r = 2$ case in the theorem statement.
\end{itemize}
\end{proof}

\mypar{Tightness of the bound} Unlike the add-remove bound, this bound is not tight in the extreme setting where $p = 1, B < n$ and truncation always happens. The reason is that the application of \cref{lem:dr} is slack: This application gives a worst case bound for distinguishing pairs of distributions $(1-q) \cdot N(\boldzero, \sigma^2 \mathbb{I}) + q \cdot N(\bfv, \sigma^2 \mathbb{I})$ and $(1-q) \cdot N(\boldzero, \sigma^2 \mathbb{I}) + q \cdot N(\bfw, \sigma^2 \mathbb{I})$ where $\ltwo{\bfv}, \ltwo{\bfw}$ are bounded. However, in our application $\bfw = \bfu - \bfv$, so we additionally know $\ltwo{\bfv - \bfw} \leq 1$, but this is information our application of \cref{lem:dr} does not use. So unlike our add-remove bound, there is some potential to improve this bound even in this extreme setting.

We conjecture that for distinguishing $(1-q) \cdot N(\boldzero, \sigma^2 \mathbb{I}) + q \cdot N(\bfv, \sigma^2 \mathbb{I})$ and $(1-q) \cdot N(\boldzero, \sigma^2 \mathbb{I}) + q \cdot N(\bfu + \bfv, \sigma^2 \mathbb{I})$ (1) a tightly dominating pair does not exist, (2) this pair is $(\epsilon, \delta)$-indistinguishable if one of $P = (1-q) \cdot N(0, \sigma^2) + q \cdot N(2, \sigma^2)$ or $P = N(0, \sigma^2)$ and $Q = (1-q) \cdot N(0, \sigma^2) + q \cdot N(1, \sigma^2)$ are $(\epsilon, \delta)$-indistinguishable, which could provide a method for near-tight PLD accounting. The improvements from this conjecture are likely to be relatively small, so we did not pursue this approach further.

\subsection{Replace-One Adjacency}

We finally extend the results to the replace-one adjacency:

\begin{theorem}\label{thm:truncation-replace}
Given $n, p, B, \sigma$ define the distributions $P, Q$ as follows:
\begin{itemize}
    \item We sample 
    $r = 1$ w.p. $ \Pr[\Binom(n-1, p) < B]$ and $r = 2$ otherwise.
    \item For $P$, we set 
    \[x = \left\{\begin{array}{lr}
        \Bern(p) & \text{if } r = 1,\\
        2 \cdot \Bern\left(p \cdot \sum_{s = B}^{n-1} \frac{\Pr[\Binom(n-1, p) = s]}{\Pr[\Binom(n-1, p) \geq B]} \cdot \frac{B}{s+1}\right) & \text{if } r = 2.\\
    \end{array}\right\}.\]
    For $Q$, we set 
    \[x = \left\{\begin{array}{lr}
        -\Bern(p) & \text{if } r = 1,\\
        -2 \cdot \Bern\left(p \cdot \sum_{s = B}^{n-1} \frac{\Pr[\Binom(n-1, p) = s]}{\Pr[\Binom(n-1, p) \geq B]} \cdot \frac{B}{s+1}\right) & \text{if } r = 2.\\
    \end{array}\right\}.\]
    \item The final random variable is $(r, x + z)$, where $z \sim N(0, \sigma^2)$.
\end{itemize}
$P, Q$ is a dominating pair for $\calM_{B,p,\sigma}$ with respect to $\stackrel{RO,n}{\sim}$.
\end{theorem}
\begin{proof}
The proof proceeds similarly to \cref{thm:truncation-zeroout}. The two key differences are (1) when we don't truncate the batch, we reduce to the known dominating pair for the Poisson sampled Gaussian mechanism under the replace adjacency (as proven in e.g. \cite{koskela20computing}) and (2) when we do truncate the batch, instead of distinguishing between $(1-q) \cdot N(-\pi_{n-1}(B), \sigma^2 \mathbb{I}) + q \cdot N(\bfx_1, \sigma^2 \mathbb{I})$ and $(1-q) \cdot N(-\pi_{n-1}(B), \sigma^2 \mathbb{I}) + q \cdot N(\boldzero, \sigma^2 \mathbb{I})$ as in the zero-out case, we are distinguishing between  $(1-q) \cdot N(-\pi_{n-1}(B), \sigma^2 \mathbb{I}) + q \cdot N(\bfx_1, \sigma^2 \mathbb{I})$ and $(1-q) \cdot N(-\pi_{n-1}(B), \sigma^2 \mathbb{I}) + q \cdot N(\bfx_1', \sigma^2 \mathbb{I})$. So in applying \cref{lem:dr}, we instead reduce to the pair of distributions $(1-q) \cdot N(0, \sigma^2) + q \cdot N(-2, \sigma^2)$ and $(1-q) \cdot N(0, \sigma^2) + q \cdot N(2, \sigma^2)$.
\end{proof}

We briefly remark that as in the zero-out case, there is some slack in the application of \cref{lem:dr} for similar reasons.

\subsection{Implementation of PLD Accounting}\label{sec:implementation}

In all our main results, we reduce to a pair of distributions where the final random variable is $(r, x+z)$. Conditioned on $r = 1$ or $r = 2$, the resulting pair of random variables is some variant of a Gaussian mechanism whose privacy loss distribution can be handled using existing accounting tools. Since $r$ is part of the output and the probability that $r = 1$ or $r = 2$ does not depend on the dataset, the privacy loss of an output $(r, x+z)$ does not change when we condition on $r$. So, we can compute the privacy loss distribution of our dominating pairs by computing the privacy loss distributions of the dominating pair conditioned on $r = 1$ or conditioned on $r = 2$, and then taking a weighted mixture of these conditional privacy loss distributions to be our final privacy loss distribution. Put more simply, we are using a ``mixture of mechanisms'' where the choice of mechanism in the mixture is made public, so the privacy loss distribution of the mixture of mechanisms is same as the corresponding mixture of the privacy loss distributions of the individual mechanisms. Concurrently with this note we have provided support for mixtures of privacy loss distributions in the \texttt{dp\_accounting} Python library \cite{dp_accounting}, as well as a \texttt{TruncatedSubsampledGaussianDpEvent} in \texttt{dp\_accounting} which incorporates the accounting provided in this note.

We also make the computation of the probability $p \cdot \sum_{s = B}^{n-1} \frac{\Pr[\Binom(n-1, p) = s]}{\Pr[\Binom(n-1, p) \geq B]} \cdot \frac{B}{s+1}$ more efficient using the following observation:

\begin{align*}
p \cdot \sum_{s = B}^{n-1} \frac{\Pr[\Binom(n-1, p) = s]}{\Pr[\Binom(n-1, p) \geq B]} \cdot \frac{B}{s+1} &= p \cdot \sum_{s = B}^{n-1} \frac{\binom{n-1}{s} p^s (1-p)^{n-s-1}}{\Pr[\Binom(n-1, p) \geq B]} \cdot \frac{B}{s+1}\\
&= p \cdot \sum_{s = B}^{n-1} \frac{\frac{(n-1)!}{(s+1)!(n-s-1)!} p^s (1-p)^{n-s-1}}{\Pr[\Binom(n-1, p) \geq B]} \cdot B\\
&= \sum_{s = B}^{n-1} \frac{\frac{(n-1)!}{(s+1)!(n-s-1)!} p^{s+1} (1-p)^{n-s-1}}{\Pr[\Binom(n-1, p) \geq B]} \cdot B\\
&= \sum_{s = B}^{n-1} \frac{\frac{n!}{(s+1)!(n-s-1)!} p^{s+1} (1-p)^{n-s-1}}{\Pr[\Binom(n-1, p) \geq B]} \cdot \frac{B}{n}\\
&= \sum_{s = B}^{n-1} \frac{\Pr[\Binom(n,p) = s+1]}{\Pr[\Binom(n-1, p) \geq B]} \cdot \frac{B}{n}\\
&= \frac{\Pr[\Binom(n,p) \geq B+1]}{\Pr[\Binom(n-1, p) \geq B]} \cdot \frac{B}{n}.\\
\end{align*}

The final expression can be evaluated in a constant number of calls to e.g. \texttt{scipy.stats.binom.sf}, whereas the initial sum requires $\approx n$ calls to \texttt{scipy} methods.

\section*{Acknowledgements}

We are thankful to Pritish Kamath and Ryan McKenna for their feedback and suggestions on improving the writeup and discussions on different adjacency definitions. Pritish additionally suggested the more efficient calculation of $\sum_{s = B}^{n-1} \frac{\Pr[\Binom(n-1, p) = s]}{\Pr[\Binom(n-1, p) \geq B]} \cdot \frac{B}{s+1}$.

\bibliographystyle{alpha}
\bibliography{ref}
\end{document}